\documentclass[11pt]{article}
\usepackage{fullpage}
\usepackage{amsfonts, amsthm, amssymb, amsmath,mathrsfs, array}
\usepackage[lined,boxed,commentsnumbered]{algorithm2e}
\usepackage[english]{babel}

\usepackage[pdftex]{graphicx,color}
\usepackage[pdftex, pagebackref=true, colorlinks=true, urlcolor=blue, citecolor=blue, linkcolor=blue]{hyperref}

\usepackage{framed}
\usepackage{verbatim}
\usepackage{fullpage}
\usepackage{epsfig}

\newtheorem{thm}{Theorem}[section]
\newtheorem{lem}[thm]{Lemma}
\newtheorem{cor}[thm]{Corollary}
\newtheorem{defn}[thm]{Definition}
\newtheorem{clm}[thm]{Claim}
\newtheorem{cons}[thm]{Construction}
\newtheorem{prop}[thm]{Proposition}

\newenvironment{theorem}{\begin{thm}\begin{rm}}%
{\end{rm}\end{thm}}
\newenvironment{lemma}{\begin{lem}\begin{rm}}%
{\end{rm}\end{lem}}
\newenvironment{corollary}{\begin{cor}\begin{rm}}%
{\end{rm}\end{cor}}
\newenvironment{definition}{\begin{defn}\begin{em}}%
{\end{em}\end{defn}}
\newenvironment{claim}{\begin{clm}\begin{rm}}%
{\end{rm}\end{clm}}
{\end{em}\end{cons}}
{\end{em}\end{prop}}

\newcommand{\secref}[1]{\hyperref[#1]{Section \ref{#1}}}
\newcommand{\apref}[1]{\hyperref[#1]{Appendix \ref{#1}}}
\newcommand{\thref}[1]{\hyperref[#1]{Theorem \ref{#1}}}
\newcommand{\defref}[1]{\hyperref[#1]{Definition \ref{#1}}}
\newcommand{\corref}[1]{\hyperref[#1]{Corollary \ref{#1}}}
\newcommand{\lemref}[1]{\hyperref[#1]{Lemma \ref{#1}}}
\newcommand{\clref}[1]{\hyperref[#1]{Claim \ref{#1}}}
\newcommand{\consref}[1]{\hyperref[#1]{Construction \ref{#1}}}
\newcommand{\figref}[1]{\hyperref[#1]{Figure \ref{#1}}}
\newcommand{\eqnref}[1]{\hyperref[#1]{Equation \ref{#1}}}

\title{
{{On Mimicking Networks Representing Minimum Terminal Cuts}} \\
}
\vspace{0.15in}
\author{
{Arindam Khan  \ \ \   Prasad Raghavendra  \ \ \ Prasad Tetali \ \ \  L{\'a}szl{\'o} A. V{\'e}gh } \\ \\
School of Computer Science, \\
Georgia Institute of Technology, \\
Atlanta, GA 30332-0765. \\
Email : {\em{akhan67@gatech.edu}}\ \ \ \ {\em{raghavendra@cc.gatech.edu}} \\ {\em{tetali@math.gatech.edu}} \ \ \ \ {\em{lvegh@cc.gatech.edu}}\\
}

\author{
  { Arindam Khan} \footnotemark[1]\ %\footnotemark[3]\
  \and { Prasad Raghavendra} \footnotemark[2]\ %\footnotemark[7]\
  \and { Prasad~Tetali} \footnotemark[1]\
  \and {L{\'a}szl{\'o} A. V{\'e}gh} \footnotemark[3]\
  }

\date{}
\begin{document}
\maketitle

\renewcommand{\thefootnote}{\fnsymbol{footnote}}

\footnotetext[1]{
School of Computer Science, 
Georgia Institute of Technology, 
Atlanta, GA 30332-0765. 
Email: {{\tt akhan67@gatech.edu}},
{{\tt tetali@math.gatech.edu}} }
\footnotetext[2]{
EECS, Univ of California, Berkeley, CA. Email: {\tt
prasad@cs.berkeley.edu}}

\footnotetext[3]{
Dept of Management, London School of Economics. Email: {\tt
L.vegh@lse.ac.uk}}

\renewcommand{\thefootnote}{\arabic{footnote}}

\begin{abstract}

Given a capacitated undirected graph $G=(V,E)$ with a set of terminals
$K \subset V$, a {\it mimicking network} is a smaller graph $H=(V_H,E_H)$ that
exactly preserves all the minimum cuts between the terminals.  Specifically,
the vertex set of the sparsifier $V_H$ contains the set of terminals $K$ and for every
bipartition $U , K-U $ of the terminals $K$, the size of the minimum cut
separating $U$ from $K-U$ in $G$ is exactly equal to the size of the
minimum cut separating $U$ from $K-U$ in $H$.

%Given a capacitated undirected graph $G=(V,E)$ with a set of terminals
%$K \subset V$, a {\it cut sparsifier} is a smaller graph  $H=(V_H,E_H)$ that
%preserves all the minimum cuts between the terminals.  Specifically,
%the vertex set of the sparsifier $V_H$ contains the set of terminals $K$ and for every
%subset $U \subset K$ of terminals, the size of the minimum cut
%separating $U$ from $K-U$ in $G$ is approximately equal to the size of the
%minimum cut separating $U$ from $K-U$ in $H$.   The approximation
%ratio between the size of the minimum cuts in $G$ and $H$ is referred to as the quality of the sparsifier.
%
%While vertex sparsifiers with $V_H = K$ have received much attention
%in literature, vertex sparsifiers with additional Steiner nodes
%(namely $K \subset V_H$) are poorly understood.  Cut sparsifiers with quality 1 are called exact cut sparsifiers or mimicking networks. 

This notion of a {\it mimicking network} was introduced by Hagerup, Katajainen, Nishimura and
Ragde \cite{HagerupKNR95} who also exhibited a mimicking network of
size $2^{2^{k}}$ for every graph with $k$ terminals.  The best known
lower bound on the size of a mimicking network is linear in the number
of terminals.  More precisely, the best known lower bound is $k+1$ for
graphs with $k$ terminals \cite{ChaudhuriSWZ00}. 

In this work, we improve both the upper and lower bounds reducing the
doubly-exponential gap between them to a single-exponential gap.
Specifically, we obtain the following upper and lower bounds on mimicking networks:
\begin{itemize}
\item 	Given a graph $G$, we exhibit a construction of mimicking network with at most $(|K|-1)$'th Dedekind number ($\approx 2^{{(k-1)} \choose {\lfloor {{(k-1)}/2} \rfloor}}$) of vertices (independent of size of $V$).  Furthermore, we show that the construction is optimal among all {\it
restricted mimicking networks} -- a natural class of mimicking
networks that are obtained by clustering vertices together. \\
%\TODO Special graphs
%\TODO Improved over Dedekind number  \\
%\TODO Lower bound for contraction based sparsifiers \\
%\TODO Specific graphs for which lower bound for contraction comes.  \\

\item There exists graphs with $k$ terminals that have no mimicking network of size
	smaller than $2^{\frac{k-1}{2}}$.
	
 %For trees with $k$ terminals we show $(k-1)$ extra Steiner vertices are sufficient to construct an exact flow sparsifier.\\
%\TODO Improvement of lower bound by showing fewer variables. Use of symmetry?

We also exhibit improved constructions of mimicking networks for
trees and graphs of bounded tree-width.
%($\frac{13k}{8}-\frac{3}{2}$) and more generally, graphs of bounded($t$)
%tree-width($|K|2^{ (3t+2) \choose {\lfloor {{(3t+2)}/2} \rfloor}}$).

%graphs and uniform-cost complete graph($|K|+1$). Also we show constructions that preserve all cuts separating terminal set of size $\le 2$ from other terminals using only one extra Steiner vertex.

\end{itemize}

\textbf{keywords:}
Approximation algorithms, Graph algorithms, Vertex sparsification, Cut sparsifier, Mimicking networks, Terminal cuts, Realizable external flow, Network flow.
\end{abstract}

\thispagestyle{empty}
\newpage

\section{Introduction}

Suppose there are small number of terminals or clients that are part of a
huge network such as the internet.  Often, it is useful to
construct a smaller graph which preserves the
properties of the huge network that are relevant to the terminals.
For example, if the terminals or clients are interested in routing
flows through the large network, one would want to construct a small
graph which preserves the routing properties of the original network.
The notion of {\it mimicking networks} introduced by Hagerup et.\ al.
\cite{HagerupKNR95} is an effort in this direction.

Let $G$ be an undirected graph with edge capacities $c_e$ for
$e \in E$, and a set of $k$ terminals $K \subset V$.  
A {\it mimicking network} for $G$ is an undirected
capacitated graph $H=(V_H,E_H)$ such that $K
\subseteq V_H$ and for each subset $U \subset K$ of terminals, the
size of the minimum cut separating $U$ from $K-U$ in $H$ is exactly
equal to the size of the minimum cut separating $U$ and $K-U$ in
the graph $G$.  As a corollary, the set of realizable external flows (possible total
flows at terminals) in $G$ are preserved in a mimicking network.
Therefore, the smaller graph $H$ {\it mimics} the graph $G$ in terms
of external flows routable through it.  The vertices of the mimicking
network that are not terminals, namely $V_H - K$ will be referred to
as {\it Steiner} vertices.

The work of Hagerup et.\ al.\cite{HagerupKNR95} exhibited a construction a
mimicking network with at most $2^{2^k}$ vertices for every graph with
$k$ terminals.  Subsequently, Chaudhuri et.\
al.\cite{ChaudhuriSWZ00} proved that there exists graphs that require
at least $(k+1)$ vertices in its mimicking network.     The same work
also obtained improved constructions of mimicking networks for special
classes of graphs namely, bounded treewidth and outer planar graphs.
Specifically, they showed that graphs of treewidth $t$ admit a
mimicking network of size $k 2^{2^{3(t+1)}}$, while outerplanar graphs
admit mimicking networks of size $(10k -6)$.     

Mimicking networks constituted the main building block in the development of $O(n)$ time algorithm for computing maximum $s-t$ flow in a bounded treewidth network \cite{HagerupKNR95} and for obtaining an optimal solution for the all-pairs minimum-cut problem in the same class of networks \cite{ArikatiCZ98}.
However, there still remained a doubly exponential gap between the
known upper and lower bounds for the size of mimicking networks for general
graphs.

\subsection{Vertex Sparsifiers}

Closely tied to mimicking networks is the more general notion of
{\it vertex sparsifiers} introduced by Moitra \cite{MoitraFocs09}.  
Roughly speaking, a {\it vertex cut sparsifier} is a mimicking network
that only approximately preserves the cut values.  Formally, let $G$ be an undirected graph with edge capacities $c_e$ for
$e \in E$, and a set of $k$ terminals $K \subset V$.  
A {\it vertex cut sparsifier} with quality $q$ is an undirected
capacitated graph $H=(V_H,E_H)$ such that $K
\subseteq V_H$ and for each subset $U \subset K$ of terminals, the
size of the minimum cut separating $U$ from $K-U$ in $H$ is within a
factor $q$ of the size of the minimum cut separating $U$ and $K-U$ in
the graph $G$.  

%The related notion of {\it vertex flow sparsifiers} introduced by Leighton and Moitra
%\cite{LeightonMStoc10} corresponds to constructing a small graph $H$
%that preserves the routing properties of the network $G$.  More
%precisely, a quality-$q$ \textit{flow sparsifier} preserves the minimum edge congestion required to route any set of demands over the terminal vertices.  
%A {\it vertex flow sparsifier} is a stronger notion than a {\it vertex
%cut sparsifier} in that, every quality $q$ flow-sparsifier is also a
%quality $q$ cut-sparsifier.

  The original motivation behind the notion of
vertex cut sparsifiers was to obtain improved approximation
algorithms for certain graph partitioning and routing problems.
If the solution to some combinatorial optimization problem only
depends on the values of the minimum cuts separating terminal subsets,
then given any approximation algorithm for the problem, we can first
compute a cut sparsifier $H$ for graph $G$ and run the approximation
algorithm on the graph $H$ instead of $G$.  If the approximation
guarantee of the algorithm depended on the number of the vertices of
the input graph, then this would yield an algorithm whose
approximation guarantee only depends on the size of the sparsifier
$H$.

The problem of constructing {\it vertex cut sparsifiers} has received
considerable attention since their introduction in
\cite{MoitraFocs09}.  Naturally, the goal would be to obtain as good an approximation as
possible, while keeping the size of the sparsifier $H$ small.  In
fact, the notion of vertex sparsifiers as defined in
\cite{MoitraFocs09} require that the graph $H$ have only the terminals
$K$ as the vertices, i.e., $V_H = K$ (no Steiner vertices).  Much of the subsequent efforts
have been focused on vertex sparsifiers with this additional
requirement that $V_H = K$.  In this setting, Moitra
\cite{MoitraFocs09} showed the existence of vertex sparsifiers with
quality $O(\log^2 k/ \log \log k)$. Subsequent works by Leighton et al.\ \cite{LeightonMStoc10}, Englert et al.\ \cite{EnglertGKRTT10} and Makarychev et al.\  \cite{MM10} gave polynomial-time algorithms for constructing $O(\log k/ \log \log k)$ cut sparsifiers, matching the best known existential upper bound. 
On the negative side, Leighton and Moitra \cite{LeightonMStoc10}
showed a lower bound of $\Omega(\log \log k)$ on the quality of cut
sparsifiers without Steiner vertices, which was subsequently improved to
$\Omega(\sqrt{ \log k/\log  \log k})$ \cite{MM10}.
%Charikar et. al. has shown the lower bound for cut sparsifiers to be $\Omega(\log^{1/4}k) $ \cite{CharikarFocs09}. 

In light of these lower bounds, it is natural to wonder if better
approximation guarantees could be obtained by vertex sparsifiers that
include {\it steiner vertices}, i.e., vertices of the sparsifier $H$ are a strict
super-set of the set of terminals $K$. % In particular, these vertex
%sparsifiers would have additional nodes apart from the terminals $K$,
%that are referred to as {\it Steiner} nodes.
%For example, we can always exactly approximate terminal cut functions using no extra Steiner vertices when $k=3$.
In fact, for $k \ge 4$, there exist graphs for which no cut sparsifier
without Steiner vertices preserves terminal cuts exactly.  But by
Hagerup \cite{HagerupKNR95},  there exists cut sparsifiers (mimicking
networks) with
$2^{2^4}$ nodes that exactly preserves all the cuts.
%which we cannot exactly approximate terminal cuts without using extra Steiner vertices \cite{LeightonMStoc10}.

Initiating the study of vertex sparsifiers with steiner nodes, Chuzhoy
\cite{Chuzhoy12} exhibited efficient algorithms to construct
$3(1+\epsilon)$-quality cut sparsifiers of size $O(C/\epsilon)^3$ for
a constant $\epsilon \in (0,1)$, where $C$ denotes the total capacity of
the edges incident on the terminals, normalized so as to make all the
edge-capacities at least $1$. The same work also gives an
efficient construction of a $(68+\epsilon)$-quality vertex flow sparsifier of size $C^{O(\log \log C)}$ in time $n^{O(\log C)}. 2^C$. 
Notice that the size of the sparsifiers depend on the total capacity
$C$ of edges incident at the terminals, which could be arbitrarily
large compared to the number of terminals $k$.  

%Specifically,  Charikar et. al.\
%\cite{CharikarFocs09} show that without Steiner vertices, there is a
%superconstant gap between the the quality of the best possible
%restricted flow sparsifier and the quality of the optimal
%flow-sparsifier.
%Englert et. al.\ \cite{EnglertGKRTT10} show a $\Omega(\sqrt{\log k})$ lower bound on the quality of sparsifiers that contain only terminal vertices and can be obtained from convex-combinations of 0-extensions in graph $G$.  

%\TODO two mimicking networks paper, minimum terminal cuts and external flows.

While there has been progress in efficient constructions of vertex
sparsifiers without Steiner nodes, the power of vertex sparsifiers with Steiner nodes is poorly
understood.  For instance, the following question originally posed by
Moitra \cite{MoitraFocs09} remains open.

\textit{Do there exists cut sparsifiers with $k^{O(1)}$ additional
steiner nodes that yield a better than $O(\log k /\log\log k)$
approximation?} 

In fact, Moitra \cite{MoitraFocs09} points out that there could exist
exact cut sparsifiers (quality $1$) with only $k$ additional Steiner
nodes.

%we will define the notion of a
%\em{contraction-based} sparsifier or a \em{contraction-based}
%mimicking network.  

\subsection{Our results:}
  In this paper, we show upper and lower bounds for {mimicking
  networks} aka vertex cut sparsifiers with quality $1$.
  First, we present an improved bound on the size of mimicking networks for
  general graphs.  Specifically, we exhibit a construction of
  mimicking networks with
at most $(|K|-1)$'th Dedekind number ($\approx 2^{{(k-1)} \choose
{\lfloor {{(k-1)}/2} \rfloor}}$) of vertices, as opposed to $2^{2^k}$
vertices.

\begin{theorem}
\label{thm:ccut}
For every graph $G$, there exists a mimicking network with
quality $1$ that has at most $(|K|-1)$'th Dedekind number ($\approx 2^{{(k-1)} \choose {\lfloor {{(k-1)}/2} \rfloor}}$)  vertices.  Further, the
mimicking can be constructed in time polynomial in $n$ and $2^k$.
\end{theorem}

We also note that the mimicking network constructed above is a
{\it contraction-based} in the sense that the mimicking network $H$ is
constructed as follows:  Fix an appropriate partition $\mathcal{C}$ of
the vertices of the graph $G$ and contract every subset of vertices $S
\in \mathcal{C}$ in the partition to form a vertex of $H$.
Contraction-based sparsifiers have also referred to as {\it restricted
sparsifiers} in literature \cite{CharikarFocs09} who show that they
are a strictly stronger notion than vertex cut sparsifiers.
For restricted sparsifiers, we will use the terms -- non-terminal and Steiner vertex interchangeably.

%The above theorem yields the first construction of vertex sparsifiers
%whose size is a function of $k$, and quality better than $O(\log
%k/\log \log k)$. 
We prove that construction is optimal
for the class of contraction-based mimicking networks.  
\begin{theorem}
\label{thm:loweRestrict}
Let $G$ be a graph with unique minimum terminal cuts. Then the
mimicking network constructed using Algorithm
\ref{alg:Exact-Cut-Sparsifier} is an optimal among contraction-based
mimicking networks for $G$ i.e., it has minimum number of vertices
among all contraction-based mimicking networks.
\end{theorem}
%\begin{comment}
%There exists graphs $G$ for which every restricted cut sparsifier with
%quality $1$ has at least $2^{2^{k-1}}$ vertices.
%\end{comment}
Next, we obtain an exponential lower bound on the size of
the mimicking networks.  Specifically, we show the following result.
\begin{theorem}
\label{thm:lower}
There exists graphs $G$ for which every mimicking network has size
at least $2^{(k-1)/2}$.
\end{theorem}
%\TODO Results for special graphs

We also obtain improved constructions of mimicking networks for
special classes of graphs like trees and graphs of bounded tree width.
For the case of a tree, we show that $\frac{13|K|}{8}-\frac{3}{2}$
suffice, while for a graph with treewidth $t$ there exists mimicking
networks of size $|K|2^{ (3t+2) \choose {\lfloor {{(3t+2)}/2}
\rfloor}}$.  We also exhibit mimicking networks that preserve cuts separating terminal set of size $\le 2$ from other terminals using only one extra Steiner vertex.

\paragraph{Related Work}  In an independent work, Krauthgamer and Rika
\cite{KrauthgamerR12} obtained upper and lower bounds for the size of
mimicking networks in general graphs, and certain special classes of
graphs.  Specifically, they show a lower bound of $2^{\Omega(k)}$ for
the size of mimicking networks even for the case of bipartite graphs.
Furthermore, the lower bound is shown to hold for the size of any data
structure that stores all the minimum terminal cut values of a graph.
The paper also obtains improved upper and lower bounds for the special
case of planar graphs.

It has been brought to our attention that the improved upper bound of
Dedekind number of vertices for mimicking networks was also observed
by Chambers and Eppstein \cite{ChambersE10}.

%%%%%%%%%%%%%%%%%%%%%%%%%%%%%%%%%%%%%%%%

\section{Preliminaries}
In this section, we set up the notation and present formal definitions
of vertex cut sparsifiers and mimicking networks.   Let $G=(V, E)$ be an undirected capacitated graph with
edge capacities $c(e)$ for all edges $e \in E$  and a set $K \subset V$ of terminals of size $k$.
Without loss of generality, we assume that $G$ is connected, otherwise
each component can be handled separately.
Let $c: E \rightarrow \mathbb{R}^+$ be the capacity function of the graph. Let $h_G : 2^V \rightarrow  \mathbb{R}^+$  denote the cut function of $G$:
\begin{equation}
h_G(A)=\sum_{e \in \delta(A)} c(e) \nonumber
\end{equation}
where $\delta(A)$ denote the set of edges crossing the cut $(A, V \setminus A)$.
Now we define terminal cut function $h_K^G : 2^K \rightarrow
\mathbb{R}^+$ on $K$ as
\begin{equation}
h_K^G(U)=min_{A\subset V , A \cap K =U}h_G(A)\nonumber
\end{equation}

In words, $h_K^G(U)$ is the cost of the minimum cut separating $U$
from $K \setminus U$ in $G$. Let $S(U)$ be the smallest subset of $V$ such that $h_G(S(U))=h_K^G(U), S(U) \cap K =U$ i.e., $S(U)$ is the partition containing $U$ in the minimum terminal cut separating $U$ from $K-U$ and if there are multiple minimum terminal cuts we take any one with minimum number of vertices  in the partition that contains $U$.
For any fixed $U \subset K$, the minimum
cut $h_K^G(U)$ can be computed efficiently.  We will sometimes abuse
the notation and use $h_K^G(U)$ to denote both the size of the minimum
cut and the set of edges belonging
to the minimum terminal cut.

If $|U|=1$, we call the minimum terminal cut separating $U$ from $K-U$ to be \textit{mono-terminal cut}. 
If $|U|\le 2$, we call the minimum terminal cut separating $U$ from $K-U$ to be \textit{bi-terminal cut}. 

\begin{definition}
$H=(V_H,E_H)$ is a cut-sparsifier for the graph $G=(V, E)$ and the
terminal set $K$, if $K \subseteq V_H$ and if the cut function $h_K^H: 2^{V_H} \rightarrow \mathbb{R}^+$ of $H$ satisfies for all $U \subset K$, 
\begin{equation}
h_K^G(U) \le h_K^H(U). \nonumber
\end{equation}
\end{definition}
Quality of cut sparsifier is a measure of how well the cut function of $H$ approximates the terminal cut function.
\begin{definition}
The quality of a cut sparsifier $H$: $Q_C(H)$ is defined as 
\begin{equation}
max_{U \subset K} h_K^H(U)/h_K^G(U). \nonumber
\end{equation}
\end{definition}
In this paper, we will study mimicking networks that are a special
class of vertex sparsifiers.
\begin{definition}
A vertex sparsifier $H$ for graph $G$ and terminal set $K$ is a
mimicking network if $Q_C(H)=1$.
\end{definition}
Nearly all existing constructions of vertex sparsifiers are based on
edge-contractions.  Now we present a simple lemma to show contraction of edges always gives us a vertex sparsifier.
\begin{lemma}
Given a graph $G$ and an edge $e$, contracting the edge $e$ in the
graph $G$ will not decrease the value of any minimum terminal cut.\cite{MoitThes}
\end{lemma}
\begin{proof} 
Let $G/e$ be the graph obtained by contracting the edge $e = (u,v)$ in
the graph $G$.  For any $U \subset K$, the minimum cut in $G/e$
separating $U$ from $K-U$ is also a cut in G separating $U$ from
$K-U$, with the additional restriction that $u$ and $v$ appear on the same side of the cut.
Thus contracting an edge (whose endpoints are not both terminals) cannot decrease the value of minimum cut separating $U$ from $K \setminus U$ for $ U \subset K$. 
\end{proof}
Vertex sparsifiers that can be obtained by contracting edges of the
original graph will be referred to as {\it contraction-based} vertex sparsifiers.
\begin{definition}
	A graph $H=(V_H, E_H)$ is a {\it contraction-based} vertex
	sparsifier/mimicking network of graph $G=(V, E)$ with terminal set $K$ if there exists a function $f: V \rightarrow V_H$ such that  the edge cost function of $H$ is defined as follow: $c_H(y, z)= \sum_{u, v|f(u)=y, f(v)=z} c(u,v)$ where $(y, z) \in E(H)$ and $(u, v) \in  E(G)$.
\end{definition}

%Any cut sparsifier with quality $\alpha$ also preserves realizable external flow pattern with quality $\alpha$. An \textit{external flow} $x=(x_1, x_2 \cdots x_k)$ is an assignment of real values $x_i$ to each terminals. A \textit{realizable external flow} is an external flow such that there exists a valid flow in which the net flow coming out of $i$'th terminal is $x_i$.

%\TODO Sparsifiers with additional property e.g. being a tree/bounded tree width/planar,  \\
%\TODO Define single, bi, tri termal cuts single+bi is done now\\
% Below commented: Flow sparsifier
\begin{comment}
Now we define flow sparsification. Assume $D$ is a demand vector that specifies the demand between every pair of terminals. A flow $F$ is a routing of demands if it satisfies demand constraints for all pairs. Congestion of an edge $e$  due to flow $F$: $Cong(e,F)$  is $F(e)/c_e$ where $F(e)$ is the flow sent along $e$ and $c_e$ is the capacity of the edge. Congestion $Cong(G,D)$ due to routing of demands $D$ is $min_{F \phantom{.} routes \phantom{.}  D} max_{e\in E(G)} Cong(e,F)$.
\begin{definition}
$H=(V_H,E_H)$ is a flow-sparsifier of quality $q$ for the graph $G=(V,
E)$ and the terminal set $K$, if $K \subseteq V_H$ and if for any set of demands $D$ over terminals, 
\begin{equation}
Cong(H,D) \le Cong(G,D) \le q \cdot Cong(H,D) . \nonumber
\end{equation}
\end{definition}
\end{comment}

%\newpage
\section{Improved Upper Bounds on Size of Mimicking Networks}

%\subsubsection{For General graphs}
In this section we construct a mimicking network for a given
undirected, capacitated graph $G=(V,E)$ with a set of terminals $K
(\subset V) := \{ v_1, v_2 \cdots v_k \}$.  Without loss of
generality, we may assume $G$ to be connected, otherwise we can consider each component separately.

%\TODO Need to refer mimicking paper and remove the proof of thm 3.2. Can keep the algorithm for better understanding of the reader.
\begin{theorem}\textit{(Restatement of Theorem \ref{thm:ccut})}
For every graph $G$, there exists a mimicking network with
quality $1$ that has at most $(|K|-1)$'th Dedekind number ($\approx 2^{{(k-1)} \choose {\lfloor {{(k-1)}/2} \rfloor}}$)  vertices.  Further, the
mimicking network can be constructed in time polynomial in $n$ and $2^k$.
\end{theorem}
\begin{proof}
First we present the algorithm \ref{alg:Exact-Cut-Sparsifier} that
constructs the mimicking network from the graph.
%%%%%%%%%%%%%%%%%%%%%%%%%%%%%%%%%%
%%%%%       Algo : Exact-Cut-Sparsifier

\IncMargin{1em}
\begin{algorithm}
\SetKwInOut{Input}{input}\SetKwInOut{Output}{output}
\SetAlgoLined
\LinesNumbered
\Input{ A capacitated undirected graph $G$, set of terminals $K \subset V$ }
\Output{ A capacitated undirected graph $H$.}
\BlankLine
Find all $2^{k-1} -1$ minimum terminal cuts using max-flow algorithm\;
Partition the graph into $2^{2^{k-1} -1}$ clusters $\mathcal{C}_1, \mathcal{C}_2 \cdots \mathcal{C}_{2^{2^{k-1} -1}}$ such that two vertices $u, v$ belong to same cluster if they appear on same side of all the minimum terminal cuts \;
Contract each non-empty cluster into single node   \;
\emph{ Return} the contracted graph $H$   \;
\caption{{ \sc{Algorithm to construct Exact-Cut-Sparsifier}}}\label{alg:Exact-Cut-Sparsifier}
\end{algorithm}\DecMargin{1em}
%%%%%%%%%%%%%%%%%%%%%%%%%%%
%%%%%%COMMENTED%%% may be put in Appendix
\begin{comment}
There exist $2^k -2$ nonempty proper subsets of $k$ terminals. The cuts $[C,\overline{C}]$ and  $[\overline{C}, C]$ are the same. Thus we need to preserve terminal cuts due to $p=2^{k-1} -1$ subsets of terminals to preserve all cuts separating a subset of terminals from the remaining terminals.
To construct the mimicking network $H$, we find a mapping $\phi : V(G) \rightarrow V(H)$. 
Consider a $p$ dimensional hypercube $H$ where $i$'th coordinate corresponds to $h_K^G(U_i)$, i.e., the minimum cut in $G$ separating $U_i$ from $K \setminus U_i$ for $i = 1, 2, \cdots p$.  
Assume $h_K^G(U_i)=(A_i, V \setminus A_i)$ i.e., the minimum terminal cut divides $V$ into two subsets $A_i$ and $V \setminus A_i$ where $U_i \subset A_i$.
For each vertex $v$ in $G$, we set $i$'th coordinate $v_i=\textbf{1}_{v_i \in A_i}$ i.e.,  $v_i=1$ if $v_i$  is in the component containing $U_i$ and else $v_i=0$. Thus for each vertex $v$ in $G$ we get a $p$-dimensional $0/1$ vector that  corresponds to a vertex of the hypercube $H$.  
So there is a cluster(may be empty) of vertices associated with each
vertex of the hypercube. We contract all vertices in the same cluster
to obtain the mimicking network $H$.
\end{comment}
%%%%%%%%%%%%%%%%%%%%%%%%
We claim that $H$ exactly preserves all minimum terminals cuts.
\begin{claim}
$H$ is a mimicking network for $G$.
\end{claim}
\textit{Proof of claim: }
Note that we are just mapping vertices of $G$ to vertices of $H$ and not deleting any edges of $G$ in $H$, thus the minimum cut value can only grow up. Hence, $h_K^G(U) \le h_K^{H}(U)$ for any $U \subset K$. But the minimum cut separating $U$ from $K-U$ in $H$ is the dictator cut parallel to the $i$'th axis. It contains only the edges of the minimum cut separating $U$ and $K-U$ in $G$. Thus $h_K^G(U) \ge h_K^{H}(U)$. Therefore we get $h_K^G(U) = h_K^{H}(U)$. $\qed$\\
We upper bound the number of vertices in $H$ by Dedekind number to complete the proof.
\end{proof}

While the algorithm creates $2^{2^{k-1}}-1$ clusters, we will argue
that by an appropriate choice of cuts many of the clusters will be
empty.  Let $N(k)$ be the number of vertices in the mimicking network
constructed by Algorithm \ref{alg:Exact-Cut-Sparsifier}, i.e., it is
the number of non-empty regions created by $2^{k-1}-1$  minimum
terminal cuts.  Here we show $N(k)$ is at most $(k-1)$th Dedekind
number. Dedekind numbers are a rapidly-growing integer sequence
defined as follows: Consider the partial order $\subseteq$ induced on the subsets
of an $n$-element set by containment.  The $n^{th}$ Dedekind number
$M(n)$ counts the number of antichains in this partial order. Equivalently, it counts monotonic Boolean functions of $n$ variables, the number of elements in a free distributive lattice with $n$ generators, or the number of abstract simplicial complexes with $n$ elements
%%%%%%%%%%%%%%%%%%%%%

For a terminal cut $[U,K-U]$ where $v_k\notin U$, let $\{S(U),
V_G-S(U)\}$ denote the partition induced by the minimum cut separating
$[U,K-U]$.  If there are multiple minimum terminal cuts we take any
one with smallest cardinality $|S(U)|$.
%Let the minimum terminal cut $[U, K-U]$ partitions the graph into $[S(U),V(G)\setminus S(U)]$ where $U \subseteq S(U)$. Cost of the minimum terminal cut $c([U,K-U)])=min_{S(U): U\subseteq S(U), K-U \cap S(U)=\phi} c(S(U)$ where $c(S(U)=\sum_{|\{u,v\}\cap S(U)|=1}c_{uv}$.   
Now let us prove two structural properties of these minimum terminal cuts.
\begin{lemma}
\label{lemma:subsetCut}
If $X \subseteq Y \subseteq K$ then $S(X) \subseteq S(Y)$.
\end{lemma}
\begin{proof}
From submodularity property of cuts we get,
\begin{eqnarray}
(h_G(S(X)) +h_G(S(Y))) &\ge& (h_G(S(X) \cup S(Y)) + h_G(S(X) \cap S(Y))) \nonumber \\
& \ge &(h_G(S(X\cup Y) + h_G(S(X\cap Y))) = (h_G(S(Y)) +h_G(S(X))) .  
\end{eqnarray}
Here the second inequality follows from the fact $h_G(S(X) \cup S(Y)) \ge h_G(S(X \cup Y))$ and $h_G(S(X) \cap S(Y)) \ge h_G(S(X \cap Y))$. Now as all inequalities are tight in (1), we get $h_G(S(X) \cup S(Y)) =  h_G(S(X \cup Y))= h_G(S(Y))$ and $ h_G(S(X) \cap S(Y)) =  h_G(S(X \cap Y))= h_G(S(X))$.  
We have $h_G(S(X) \cap S(Y)) = h_G(S(X))$, but recall that among all
minimum cuts separating $(X, K-X)$, $S(X)$ has the smallest
cardinality. This implies $S(X) \subseteq S(Y)$.
\end{proof}

\begin{lemma}
\label{lemma:disjointCut}
If $X \cap Y =  \phi$ then $S(X) \cap S(Y)=\phi$.
\end{lemma}
\begin{proof}
Assume $S(X) \cap S(Y)\neq \phi$. Then $ h_G(S(X) \setminus S(Y)) +  h_G(S(Y) \setminus S(X))  \le ( h_G(S(X))+  h_G(S(Y))$. On the other hand as we always take the minimum terminal cut with smallest $|S(X)|$. Hence $ h_G(S(X)) <  h_G(S(X) \setminus S(Y))$ and $ h_G(S(Y)) <  h_G(S(Y) \setminus S(X))$. This contradicts.
\end{proof}

Note that each region created by algorithm
\ref{alg:Exact-Cut-Sparsifier}, is basically intersection of
partitions containing $S(X)$ for some minimum terminal cuts $(X,K-X)$ and
complement of  $S(X)$ for remaining minimum terminal cuts.
Let  $X \subseteq \{ U \subset K, v_k \notin U \}$ i.e., $X$ is a
collection of subsets of $K$ that do not contain $v_k$. Let us define
$A(X)= (\cap_{Z \in X} S(Z)) \cap (\cap_{W \notin X}
\overline{S(W)})$. Each $A(X)$ corresponds to a cluster produced by
the algorithm.  We will show that $A(X)$ is empty for many choices of
$X$.
%We prove an upset property for these sets.

\begin{lemma}
\label{lemma:upsetLemma}
If $A(X) \neq \phi$ then $X$ is upward closed set i.e., ($\forall P \in X, P \subseteq Q \Rightarrow Q \in X$). 
\end{lemma} 
\begin{proof}
Suppose there exists a $Q \notin X$ such that for some $P \in X$ and $Q \supseteq P$.
From lemma \ref{lemma:subsetCut}, 
\begin{equation}
S(P) \subseteq S(Q).
\end{equation}
Also, note that by definition, $A(X) \subseteq S(P) \cap
\overline{S(Q)}$.  Hence, we get $A(X) \subseteq S(P) \cap
\overline{S(Q)} = \phi$ -- a contradiction.
\end{proof}

%\begin{theorem}
%\label{thm:Dedekind}
%Number of regions is at most the number of antichains of $(k-1)$ elements, with all subsets sharing atleast one common element.
%\end{theorem} 

%\begin{proof}
From lemma \ref{lemma:upsetLemma}, if $A(X) \neq \phi$ then $X$ is upward closed set. Now minimal elements of upper sets form an antichain. So $N(k)$ is upper bounded by the number of antichains of subsets of an $(k-1)$-element set i.e., $M(k-1)$. Kleitman and Markowsky\cite{KM75} had shown that:
\begin{equation}
{n \choose {\lfloor {n/2} \rfloor}} \le \log M(n) \le {n \choose {\lfloor {n/2} \rfloor}}(1+O(\log n/n))
\end{equation}
Moreover from lemma \ref{lemma:disjointCut}, if there are two completely disjoint elements in $X$ that will lead to an empty region. So $N(k)$ is upperbounded by the number of antichains of subsets of $(k-1)$-element sets where all members of the antichain share at least one common element. Let us call this number to be $Z(k-1)$. Clearly $M(k-2) \le Z(k-1) \le M(k-1)$.
%lower bound of R(n)
Table \ref{table:nonlin} compares different bounds.
\begin{table}[ht]
\caption{Different bounds related to N(k)} % title of Table
\centering  % used for centering table
\begin{tabular}{c c c c c c} % centered columns (5 columns)
\hline\hline                        %inserts double horizontal lines
$k$ & Lower & Best Upper  & Upper Bound & $(k-1)$th & $2^{2^{k-1}}-1$\\ [0.5ex] 
 &bound & bound & from Contraction  & Dedekind No. &  \\[0.5ex]
&  &   & $Z(k-1)$ & $M(k-1)$ &  \\[0.5ex]
% inserts table 
%heading
\hline                  % inserts single horizontal line
2 & 2 & 2 &  2 & 2  & 3\\ % inserting body of the table
3 & 3 & 3 & 4  & 5  &  15\\
4 & 5 & 5 & 11 & 19 & 255\\
5 & 6 & 6 & 54 & 167 &  $65535$\\ 
6 & 9 & * &687 & 7580 &  $4.29 \times 10^9$\\ [1ex]      % [1ex] adds vertical space
\hline %inserts single line
\end{tabular}
\label{table:nonlin} % is used to refer this table in the text
\end{table}

%%%%%%%%%%%%%%%%%%%%%%%%%%
\begin{comment}
\begin{table}[ht]
\caption{Different bounds related to N(k)} % title of Table
\centering  % used for centering table
\begin{tabular}{c c c c c c} % centered columns (5 columns)
\hline\hline                        %inserts double horizontal lines
$k$ & Lower & Best Upper  & Upper Bound & $(k-1)$th & $2^{2^{k-1}}-1$\\ [0.5ex] 
 &bound & bound & from Contraction  & Dedekind No. &  \\[0.5ex]
&  &   & $Z(k-1)$ & $M(k-1)$ &  \\[0.5ex]
% inserts table 
%heading
\hline                  % inserts single horizontal line
$k$ & 2 & 3 &  4 & 5  & 6\\ % inserting body of the table
lower bound & 2 & 3 & 5  & 6  &  9\\
Best upper bound & 2 & 3 & 5 & 6 & *\\
Best upper bound from contractions $Z(k-1)$ & 2 & 4 & 11 & 54 &  $687$\\ 
Dedekind Numbers $M(k-1)$ & 2 & 5 & 19 & 167 &  7580\\
$2^{2^{k-1}}-1$& 3 & 15 & 255 & 65535 &  $4.29 \times 10^9$\\ [1ex]      % [1ex] adds vertical space
\hline %inserts single line
\end{tabular}
\label{table:nonlin} % is used to refer this table in the text
\end{table}
\end{comment}
%%%%%%%%%%%%%%%%%%%%%%%%%%%%%%%%%%
%\end{proof}

%\TODO Do we need to define tree-width here?\\

The observations made in this section together with results on bounded treewidth on \cite{ChaudhuriSWZ00} implies improved bound for graphs with bounded treewidth.
\begin{corollary}
Let $G$ be a $n$-vertex network of treewidth $t$. Then we can create
an mimicking network for $G$ that has size at most $k 2^{3(t+1) \choose {\lfloor {3(t+1)/2} \rfloor}}$.
\end{corollary}

\subsection{Contraction-Based Mimicking Networks}
Here we will show that on every graph $G$ that has unique
minimum terminal cuts, Algorithm
\ref{alg:Exact-Cut-Sparsifier} produces a mimicking network that is
optimal among all contraction-based mimicking networks.

\begin{theorem}
\textit{(Restatement of Theorem \ref{thm:loweRestrict})}
Let $G$ be a graph with unique minimum terminal cuts. Then the
mimicking network constructed using Algorithm
\ref{alg:Exact-Cut-Sparsifier} is optimal among contraction-based
mimicking networks for $G$ i.e., it has minimum number of vertices
among all contraction-based mimicking networks.
\end{theorem}

\begin{proof}
Let $H$ be the contraction-based mimicking network for graph $G$ with terminal set $K$ constructed using function $\phi: V(G) \rightarrow V_H$ in Algorithm \ref{alg:Exact-Cut-Sparsifier}.
First, we claim that all edges in $H$ belong to some minimum terminal cut in $G$.
\begin{claim}
\label{claim:edge}
For all edges $(y,z) \in G$, $\phi(y) \neq \phi(z)$ if and only if
$(y,z) \in h_K^G(U)$ for some $U \subset K$.
%or all edges $(u,v) \in E_H$, if $y,z \in V(G)$ such that $\phi(y)=u$, $\phi(z)=v$ then  $(y,z) \in h_K^G(U)$ for some $ U \subset K$.
\end{claim}
%\textit{Proof of Claim \ref{claim:edge}:}
\begin{proof}
The claim is clear from the construction presented in Algorithm
\ref{alg:Exact-Cut-Sparsifier}.  Two vertices are merged if and only
if the edge between them does not belong to any minimum cut.
%Take $(u,v) \in E_H$. Let $\mathcal{C}_u$ and $\mathcal{C}_{v}$
%correspond to the clusters of vertices in $G$ that are contracted to
%$u$ and $v$ respectively. Then there is some subset of terminals $U$
%such that $\mathcal{C}_u$ and $\mathcal{C}_v$ fall on different sides
%of the minimum cut separating $U$ and $K-U$.  Therefore all edges
%between  $\mathcal{C}_u$ and $\mathcal{C}_{v}$ belong to the minimum
%terminal cut separating $U$ and $K-U$.
%$\qed$
\end{proof}

Assume $H'$ is the optimal contraction-based mimicking network with
minimum number of vertices, i.e., 
$|V(H')| \le |V_H|$.  Since $H'$ is contraction-based, it is
defined by a function $\phi' : V(G) \rightarrow V_{H'}$.

\begin{claim}
\label{claim:edge1}
For all edges $(y,z) \in G$,  if $(y,z) \in h_K^G(U)$ for some $U
\subset K$ then  $\phi'(y) \neq \phi'(z)$
%or all edges $(u,v) \in E_H$, if $y,z \in V(G)$ such that $\phi(y)=u$, $\phi(z)=v$ then  $(y,z) \in h_K^G(U)$ for some $ U \subset K$.
\end{claim}
%\textit{Proof of Claim \ref{claim:edge1}:}
%Using the following two claims we will show $|V(H')| \ge |V_H|$ to complete the proof.
\begin{proof}
Consider an $e = (y,z)$ in the original graph $G$, that
belongs to some minimum terminal cut $(U,K-U)$.  We claim that the
clusters containing $y$ and $z$ are distinct in $H'$. 

By definition of $H'$, the minimum cut $h_{K}^{H'}(U)$ has
the same value as the minimum terminal cut $h_K^G(U)$.  Since all
minimum terminal cuts in $G$ are unique, this implies that the cut
induced by $h_K^{H'}(U)$ in $G$ is exactly the same as $h_K^G(U)$.
Therefore, for every edge $(y,z)$ in the graph $G$ that belongs
to a minimum terminal cut $h_K^G(U)$, the corresponding clusters in
$H'$ are distinct.
%$\qed$

\end{proof}
From the previous two claims, $\phi(y) \neq \phi(z) \implies \phi'(y)
\neq \phi'(z)$ for every edge $(y,z) \in G$.  This implies that the
number of clusters in $H$ is at most the number of clusters in $H$.
\end{proof}

\section{Exponential Lower bound}

In this section we will exhibit lower bounds on the size of mimicking
networks using a subtle rank argument.  Fix $p =
2^{k-1}-1$ for the remainder of the section.

\begin{definition}
A \textit{minimum terminal cut vector(MTCV)} $m^{G,K}$ for graph $G$ with terminal set $K$ is a $p$-dimensional vector where $i$'th coordinate  $m_i^{G,K}=h_K^G(U_i)$ i.e., it corresponds to the value of terminal cut separating $i$'th subsets of terminals from rest of the terminals for $i \in \{1, 2, \cdots p(=2^{k-1}-1)\}$. 
\end{definition}
% Note that scaling all edges by some factor does not affect the quality. So for each graph we can scale all the edges by some factor such that the maximum among the $p$ minimum terminal cuts has value less than 1. 
%By appropriate scaling, we can assume all MTCVs lie in $[0,1)^p$. 
Let $M_k$ be the set of all possible minimum terminal cut vectors with
$k$ terminals.  Not all vectors $v
\in \mathbb{R}^{2^{k-1}-1}$ can be minimum terminal cut vectors.  The
submodularity of the cut function introduces constraints on the
coordinates of the minimum terminal cut vector.
For example there are 3 possible terminal cuts for graphs with terminal set size 3. However  \textbf{[0.1, 0.1, 0.8]}  is not a valid MTCV. 
%However we will show that their range is some nonzero finite volume.
 First we prove that these minimum terminal cut vectors form a convex set.
\begin{lemma}
\label{lemma:cutcomb}
$M_k$ is a convex cone in $\mathbb{R}^{2^{k-1}-1}$.
\end{lemma}
\begin{proof}
Note that by scaling the edges of a graph $G$, the corresponding
minimum terminal cut vector also scales.  Therefore, it is sufficient
to show the convexity of the set $M_k$.

Let $G_1$ and $G_2$ be graphs with terminal set $K$ of size $k$.   Let $N_1$ and $N_2$ be their set of non-terminals respectively i.e., $N_i \cup K= V(G_i)$ for $i=1,2$. Note that these graphs might have different edge weights or different number of vertices. So depending on the edge values minimum terminal cuts will have different values. 
Let us assume $t_1$ and  $t_2 $ be the minimum terminal cut vectors for graphs  $G_1$ and $G_2$   with  same terminal set $K$ and non negative edge cost functions  $\mathcal{C}_1$ and  $\mathcal{C}_2$ respectively. 
We claim that for any nonnegative  $\lambda_1, \lambda_2$ such that $ \lambda_1 + \lambda_2=1 $, there exists a graph $H$ with same terminal set $K$ and edge cost function $\mathcal{C}'$ such that its minimum terminal cut vector $t' = \lambda_1 t_1 + \lambda_2 t_2$.
We create $H$ with nonterminals $N_1 \cup N_2$. We start with all edge costs in $H$ to be 0. 
Then for $i=1$ and $2$, for all edges $(u,v) \in G_i$, we increase the cost of edge $(u,v)$ in $H$ by $\lambda_i \mathcal{C}_i(u,v)$. The final graph has a minimum terminal cut vector of value $\sum _{i=1}^2 \lambda_i t_i$.
\end{proof}
%In fact the above lemma is true for approximately preserving minimum cuts i.e., convex combination of terminal cut vectors approximating corresponding minimum cut vectors by some factor $\alpha$ gives another terminal cut vector approximating some minimum terminal cut vector by the same factor $\alpha$.

Now we show the central lemma regarding the range of the minimum terminal cut vectors.
\begin{lemma}
\label{lemma:cutrange}
The set $M_k$ has  nonzero volume.
\end{lemma}
\begin{proof}
The \textbf{0} vector is MTCV for a completely disconnected graph. 
For each $i \in \{1,\ldots,2^{k-1}-1\}$, we will show that a line segment
in the $i^{th}$ direction belongs to $M_k$. 
By the convexity of the set $M_k$ (lemma \ref{lemma:cutcomb}) this
will imply that the set $M_k$ has nonzero volume, i.e., full
dimensional.

To demonstrate a line segment along direction $i \in \{1,\ldots
2^{k-1}-1\}$, we will show that there
exist two MTCVs which differ only in $i$'th coordinate and same in all
other $p-1$ coordinates.  Fix a subset $U_i$ of terminals.  To
construct MTCVs that differ only on the $i^{th}$ coordinate, construct a graph $H_i$ for terminal sets $U_i$ as shown in Fig. 1. 
\begin{figure}[h!]
\label{figure:graph}
  \centering
      \includegraphics[width=0.5\textwidth]{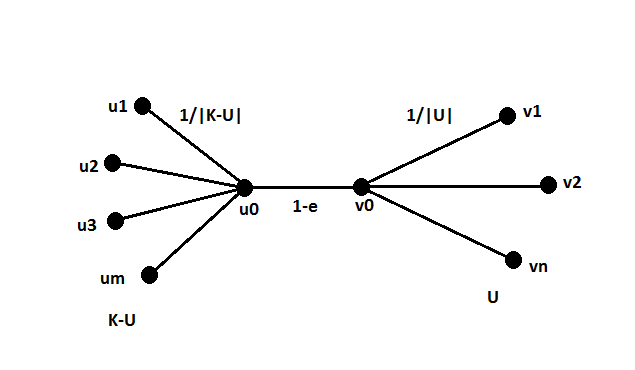}
  \caption{Graph corresponding to terminal cut $[U, K_U]$}
\end{figure}
Add all terminals in $K-U_i$ to a non-terminal $u_0$ with edge costs $1/|K-U_i|$. Add all terminals in $U_i$ to another non-terminal $v_0$ with edge costs $1/|U_i|$. Put an edge between $u_0$ and $v_0$ with edge cost $1 -\epsilon$ where $0<\epsilon< min \{1/|U_i|, 1/|K-U_i \}$. So, value of minimum terminal cut separating $U_i$ from $K-U_i$ is $1 -\epsilon$ and it contains only the edge $(u_0, v_0)$.  All other terminal cuts have value $\le 1$ and does not  contain the edge $(u_0, v_0)$.
So, we can change value of $\epsilon$ between 0 and $min \{1/|U_i|,
1/|K-U_i \}$ to obtain a line segment contained in $M_k$ along
direction $i$.

%As we can do this for all co-ordinates and $M_k$ is convex, we conclude the range to be some nonzero finite volume in $[0,1)^p$.
\end{proof}
%%%%%%%%%%%%%%%%%%%%%%%%%%%%%%%%%%%
\begin{definition}
For a given graph $G$ with terminal set $K$, the cut matrix $S_G$ is a $p \times |E(G)|$ matrix  where $S_{ij}=1$ if edge $e_j \in h_K^G(U_i)$ and 0 otherwise.
\end{definition}

\begin{theorem}
\textit{(Restatement of Theorem \ref{thm:lower})}
There exists graphs $G$ for which every mimicking network has size
at least $2^{(k-1)/2}$. 
\end{theorem}

\begin{proof}
Suppose every graph $G$ with $k$ terminals has a mimicking network
with $t$ vertices.

Consider a mimicking network $H$ with $t$ vertices for a graph $G$
with $k$ terminals.  There are
$2^{t}-1$ possible cuts in the graph $H$.  Therefore, there are at
most $(2^{t}-1)^p$ different cut matrices $S_H$ of $H$.  The specific
cut matrix $S_H$ depends on the weights of the edges in $H$.

Let us refer to these matrices as $S_1, S_2,\ldots, S_{(2^{t}-1)^p}$.
Each matrix $S_i$ can be thought of as a linear map $S_i :
\mathbb{R}^{\binom{t}{2}} \to \mathbb{R}^{2^{k-1}-1}$.
For every graph $G$, there exists a choice of weights $w_{ij}$ for the
edges of $H$, and a choice of cut matrix $S_{\ell}$ (determined by the
weights), such that $S_\ell w$ is equal to the minimum terminal cut
vector $h_{K}^G$ of the graph $G$.  Therefore, the set $M_k$ of all
MTCVs is in the union of the ranges of the linear maps $\{S_i\}_{i
=1}^{(2^{t}-1)^p}$.

However, since $M_k$ has non-zero volume (is of full dimension), at least one of the linear maps
$S_i$ must have rank $ = 2^{k-1}-1$.  Therefore $\binom{t}{2} \geq 2^{k-1}-1$
implies that $t \geq 2^{(k-1)/2}$.

\end{proof}

\begin{corollary}
There exists graphs $G$ for which every cut sparsifier that preserves $C$ minimum terminal cuts exactly has size
at least $|C|^{1/2}$. 
\end{corollary}

As the graph constructed in the theorem \ref{thm:lower} has tree-width $(k+1)$, we get the following corollary.
\begin{corollary}
There exists graphs $G$ with treewidth $\ge (k+1)$ for which every
mimicking network has size
at least $2^{(k-1)/2}$. 
\end{corollary}

\clearpage

\bibliographystyle{alpha}
\bibliography{bib-sparse}

\newcommand{\etalchar}[1]{$^{#1}$}
\begin{thebibliography}{EGK{\etalchar{+}}10}

\bibitem[ACZ98]{ArikatiCZ98}
Srinivasa~Rao Arikati, Shiva Chaudhuri, and Christos~D. Zaroliagis.
\newblock All-pairs min-cut in sparse networks.
\newblock {\em J. Algorithms}, 29(1):82--110, 1998.

\bibitem[CE10]{ChambersE10}
Erin~W. Chambers and David Eppstein.
\newblock Flows in one-crossing-minor-free graphs.
\newblock In {\em ISAAC (1)}, pages 241--252, 2010.

\bibitem[Chu12]{Chuzhoy12}
Julia Chuzhoy.
\newblock On vertex sparsifiers with steiner nodes.
\newblock {\em 44th ACM Symposium on Theory of Computing (STOC)}, 2012.

\bibitem[CLLM10]{CharikarFocs09}
Moses Charikar, Tom Leighton, Shi Li, and Ankur Moitra.
\newblock Vertex sparsifiers and abstract rounding algorithms.
\newblock {\em 51st Annual IEEE Symposium on Foundations of Computer Science
  (FOCS)}, pages 265--274, 2010.

\bibitem[CSWZ00]{ChaudhuriSWZ00}
Shiva Chaudhuri, K.~V. Subrahmanyam, Frank Wagner, and Christos~D. Zaroliagis.
\newblock Computing mimicking networks.
\newblock {\em Algorithmica}, 26(1):31--49, 2000.

\bibitem[EGK{\etalchar{+}}10]{EnglertGKRTT10}
Matthias Englert, Anupam Gupta, Robert Krauthgamer, Harald R{\"a}cke, Inbal
  Talgam-Cohen, and Kunal Talwar.
\newblock Vertex sparsifiers: New results from old techniques.
\newblock {\em APPROX-RANDOM}, pages 152--165, 2010.

\bibitem[HKNR95]{HagerupKNR95}
Torben Hagerup, Jyrki Katajainen, Naomi Nishimura, and Prabhakar Ragde.
\newblock Characterizations of k-terminal flow networks and computing network
  flows in partial k-trees.
\newblock pages 641--649, 1995.

\bibitem[KM]{KM75}
D.~Kleitman and G.~Markowsky.
\newblock On dedekind's problem: the number of isotone boolean functions: Ii.
\newblock {\em Transactions of the American Mathematical Society}, 213.

\bibitem[KR12]{KrauthgamerR12}
Robert Krauthgamer and Inbal Rika.
\newblock Mimicking networks and succinct representations of terminal cuts.
\newblock {\em Personal Communication}, 2012.

\bibitem[LM10]{LeightonMStoc10}
Tom Leighton and Ankur Moitra.
\newblock Extensions and limits to vertex sparsification.
\newblock {\em 42nd ACM Symposium on Theory of Computing (STOC)}, pages 47--56,
  2010.

\bibitem[MM10]{MM10}
Konstantin Makarychev and Yury Makarychev.
\newblock Metric extension operators, vertex sparsifiers and lipschitz
  extendability.
\newblock {\em 51st Annual IEEE Symposium on Foundations of Computer Science
  (FOCS)}, pages 255--264, 2010.

\bibitem[Moi]{MoitThes}
Ankur Moitra.
\newblock Vertex sparsification and universal rounding algorithms.
\newblock {\em PhD Thesis, MIT, 2011.}

\bibitem[Moi09]{MoitraFocs09}
Ankur Moitra.
\newblock Approximation algorithms for multicommodity-type problems with
  guarantees independent of the graph size.
\newblock {\em 50th Annual IEEE Symposium on Foundations of Computer Science
  (FOCS)}, pages 3--12, 2009.

\end{thebibliography}

%\newpage
%\vspace{0.15in}
%\noindent {\bf{\LARGE{Appendix}}}
\appendix

\section{Improved Constructions for Special Classes of Graphs}

%\subsection{Construction of an optimal restricted exact cut sparsifier}
%Let us start with a simple case when $G$ is a tree. Then we provide the sparsification for general graphs.
\subsection{Trees}
\begin{theorem}
\label{thm:uppertree}
Given an undirected, capacitated tree $T=(V,E)$ and a set $K \subset
V$ of terminals of size $k$, we can construct a mimicking network
$T_H=(V_H,E_H)$ for which the cut-function exactly approximates the
value of {\em{every}} minimum cut separating any subset $U$ of
terminals from the remaining terminals $K-U$ where $|V_H| \le 2k-1$
and this is tight for contraction-based mimicking networks. We can
also create an outerplanar mimicking network which has at most $\frac{13k}{8}-\frac{3}{2}$ vertices.

\end{theorem}
\begin{proof}
Let $H'$ be the smallest sized mimicking network. We can assume each non-leaf non-terminal vertex in $H'$ has degree at least 3. Otherwise, if nonterminal vertex $v$ is a degree 2 vertex with neighbor $u$ and $w$, then we can delete $v$ and add an edge $(u,w)$ with cost $min(c(v,u), c(v,w))$ to preserve the minimum terminal cuts. In other words, we can contract the minimum capacity edge among $(v,u)$ and $(v,w)$. Similarly if a nonterminal is a leaf, we can delete that nonterminal as it does not affect any minimum terminal cuts. 
Therefore, finally the tree $T'$ contains only terminals as leaves and each non-leaf vertex has degree at least 3. So at most there are $(2k-2)$ vertices.

To show this is tight for contraction-based mimicking network, consider a 3-regular tree with uniform edge costs and leaves as terminals. Each edge $e$ is  in at least one unique minimum terminal cut $C_e$. To preserve cut $C_e$, we can not contract $e$. Thus we need at least $(2k-3)$ edges in this case.

Now we add appropriate 0-cost edges(if needed) in $T'$ to make the tree 3-regular and set of terminals as set of leaves. We call this tree $T$.
We can rearrange the tree such that for any node $v$ in tree $T$, height of the subtree rooted at left child of $v$ is greater than the height of the subtree rooted at right child of $v$. 
Now we define an operation called $(Y$-$\Delta)$-transformation which
reduces the number of vertices further. However the mimicking network
remains no more contraction-based. Let $x$ be a degree-3 nonterminal with neighbors $u,v,w$, then we can delete $x$ and add edges $(u,v),(v,w),(w,u)$ with edge cost $\frac{c(u,x)+c(v,x)-c(w,x)}{2}, \frac{c(v,x)+c(w,x)-c(u,x)}{2}, \frac{c(u,x)+c(w,x)-c(v,x)}{2}$ respectively. We call this $(Y$-$\Delta)$-transformation. We consider non-terminals one by one in an in-order traversal of $T$. 
We apply the transformation if a vertex has degree-3 and modify the graph. 
Then we find the next vertex in the in-order traversal of $T$ that has
degree 3 in the modified graph. If there exists such a vertex, we
continue applying the transformation on it. Otherwise we stop to get
the mimicking network $H$. Note that $H$ is a cactus graph i.e., two cycles share at most one vertex in the graph. This is also an outerplanar graph. Assume $V(H)=n$.
Now we claim that there are at most $\lfloor k/2 \rfloor$ leaves in $H$. Consider the leaves(terminals) in the in-order traversal $v_1, v_2, \cdots v_k$. Pair $(v_i, v_i+1)$ for $i=1, 2, \cdots \lfloor k/2 \rfloor$. We claim that at most one of them is a leaf after completion of  $(Y$-$\Delta)$-transformation. Take the path from $v_i$ to $v_{i+1}$ in $T$. At least one degree-3 nonterminals $v_t$ is on the path such that $(Y$-$\Delta)$-transformation was applied to $v_t$, making one leaf in $T$ to have degree $\ge2$ in $H$. Also note than $v_1$ and $v_2$ both are leaves due to the arrangement.
So $H$ has at most $k/2$ leaves and at least $(n-k)$ nodes of degree 4 or more.
As $(Y$-$\Delta)$-transformation keeps number of edges same. $H$ still has at most $(2k-3)$ edges.
Thus we get, $4(n-k)+2\frac{k}{2}+\frac{k}{2} \le 2(2k-3)$ i.e., $n \le \frac{13k}{8}-\frac{3}{2}$.
\end{proof}

\end{document}